\documentclass[11pt]{article}
\usepackage[T1]{fontenc}
\usepackage{lmodern}
\usepackage[a4paper,margin=29mm]{geometry}
\usepackage{amsmath,amssymb,amsthm,mathtools}
\usepackage{booktabs,array,microtype,xcolor,needspace}
\usepackage[colorlinks=true,linkcolor=blue!45!black,citecolor=blue!45!black,urlcolor=blue!45!black]{hyperref}
\hypersetup{pdftitle={Dimension Amplification for Tarski Fixed-Point Query Lower Bounds},pdfauthor={Boyu Liu and Zihe Wang}}
\newtheorem{theorem}{Theorem}[section]
\newtheorem{lemma}[theorem]{Lemma}
\newtheorem{corollary}[theorem]{Corollary}
\newtheorem{proposition}[theorem]{Proposition}
\theoremstyle{definition}
\newtheorem{definition}[theorem]{Definition}
\theoremstyle{remark}

\DeclareMathOperator{\Fix}{Fix}
\newcommand{\one}{\mathbf{1}}
\newcommand{\HH}{\mathcal H}
\newcommand{\OO}{\mathcal O}
\newcommand{\ceil}[1]{\left\lceil #1\right\rceil}
\newcommand{\floor}[1]{\left\lfloor #1\right\rfloor}

\numberwithin{equation}{section}
\title{Dimension Amplification for Tarski\\Fixed-Point Query Lower Bounds}
\author{%
Boyu Liu\thanks{Gaoling School of Artificial Intelligence, Renmin University of China.
E-mail: \texttt{liuboyu03@ruc.edu.cn}.}
\and
Zihe Wang\thanks{Gaoling School of Artificial Intelligence, Renmin University of China.
E-mail: \texttt{wang.zihe@ruc.edu.cn}.}}
\date{October 2026}
\begin{document}
\maketitle
\begin{abstract}
We prove that finding a fixed point of a monotone map on the nine-dimensional grid
\([N]^9\) requires \(\Omega((\log N)^3)\) deterministic queries, even when the fixed
point is unique and each query returns the entire function value. The proof gives a construction
that raises the dimension from \(d\) to \(4d+1\), preserves uniqueness, and
adds a logarithmic factor to the lower bound. Iteration gives
\(\Omega((\log N)^{r+2})\) queries in dimension
\((7\cdot4^r-1)/3\), for every fixed nonnegative integer \(r\),
with an implicit constant that may depend on \(r\).
Consequently, no finite logarithmic exponent bounds the query complexity in all
fixed dimensions.
\end{abstract}

\section{Introduction}

Equilibrium computation seeks to understand the algorithmic content of
existence theorems. A theorem may guarantee a stable outcome without
providing an efficient way to locate one. For equilibria in games with
strategic complementarities, an important source of such guarantees is
Tarski's fixed-point theorem: every monotone self-map of a complete
lattice has a fixed point \cite{tarski,epry}. Understanding the
algorithmic content of this guarantee leads to the problem of finding a
Tarski fixed point.

In supermodular games, strategic complementarities mean that higher
actions by other players make higher actions more attractive. Suitable
selections of best responses form a monotone map, whose fixed points are
pure Nash equilibria. Etessami, Papadimitriou, Rubinstein, and Yannakakis
\cite{epry} established computational reductions between Tarski
fixed-point search and equilibrium computation in supermodular games,
and also gave reductions from stochastic games to Tarski. These connections
motivate a basic question: how much information about a monotone map is
needed to find a fixed point?

We study this question on the finite grid \([N]^d=\{1,\ldots,N\}^d\),
ordered coordinatewise. A map \(F:[N]^d\to[N]^d\) is monotone if
increasing input coordinates cannot decrease any output coordinate.
The task is to find a point \(x\) with \(F(x)=x\), given access to an
oracle that returns the entire vector \(F(x)\) on a query at \(x\).
This full-value model measures the information required by an algorithm,
with no restriction on computation between queries.

The two parameters are the dimension \(d\) and the side length \(N\).
Encoding each coordinate by its index in \(\{0,\ldots,N-1\}\) represents
a grid point with \(d\lceil\log_2 N\rceil\) bits. For fixed \(d\), a
query bound polynomial in \(\log N\) is therefore polynomial in the
encoding length of a query point. Iteration from the all-ones vector gives
an \(O(dN)\)-query algorithm, while recursive binary search finds a fixed point in
\(O((\log N)^d)\) queries for every fixed \(d\)
\cite{dqy}.\footnote{In bounds stated for a fixed dimension, the implicit
constants and the minimum side length for which the bound holds may
depend on the dimension, but not on the side length.}
Yet the exponent need not increase with every additional coordinate.

The tight complexity is
\(\Theta((\log N)^2)\) in dimensions two, three, and four
\cite{epry,fps,cly2026}. This raises the possibility that dimension might
affect only a multiplicative constant, rather than the exponent.
Chen and Li \cite[Section~5]{chenli} asked whether a cubic-logarithmic
lower bound holds in any constant dimension. More broadly, Fearnley et al.\
\cite{fps} and Chen, Li, and Yannakakis \cite[Section~9]{cly2026}
discussed the possibility of an \(O(f(d)(\log N)^2)\) upper bound.
We show that, for deterministic algorithms, no fixed finite exponent can
suffice in all fixed dimensions.

\subsection{Our results}

We prove a deterministic \(\Omega((\log N)^3)\) lower bound in dimension
nine, answering the cubic-logarithmic question in the deterministic
model. More generally, we construct instances requiring arbitrarily large
powers of \(\log N\) in suitable fixed dimensions. All these instances
have a unique fixed point, which allows the construction to be applied repeatedly.

For the formal statements, algorithms are deterministic and may choose
queries adaptively. An algorithm succeeds when it outputs a point \(x\)
with \(F(x)=x\); it may output a point that it has not queried.
Let \(Q_d(N)\) be the minimum worst-case number of queries for monotone
maps on \([N]^d\), and let \(U_d(N)\) be the corresponding complexity
under the promise that the fixed point is unique. Thus
\[
U_d(N)\le Q_d(N).
\]
Section~\ref{sec:model} formally defines the input families, oracle
algorithms, and query measures. Every logarithm in this paper is to base two.

\begin{theorem}[Lower-bound hierarchy]\label{thm:hierarchy}
For each fixed integer \(r\ge0\), put
\[
d_r=\frac{7\cdot4^r-1}{3}.
\]
There are constants \(c_r>0\) and \(N_r\) such that, for every integer
\(N\ge N_r\),
\[
U_{d_r}(N)\ge c_r(\log N)^{r+2}.
\]
In particular, \(U_9(N)=\Omega((\log N)^3)\).
\end{theorem}

The dimensions \(2,9,37,149\) give logarithmic exponents \(2,3,4,5\),
respectively. Each dimension is fixed before \(N\) tends to infinity;
no uniformity in the constants \(c_r\) is needed. Consequently, increasing
the dimension cannot be accounted for by a multiplicative factor while
retaining a fixed finite power of \(\log N\).

\begin{corollary}[No dimension-independent exponent]\label{cor:no-exponent}
There is no finite real number \(a\) for which every fixed dimension \(d\)
satisfies \(Q_d(N)=O((\log N)^a)\), even when the implicit constant and
the asymptotic threshold are allowed to depend on \(d\).
\end{corollary}

The central ingredient is a dimension-amplification theorem. To state it,
we distinguish outputting a fixed point from actually querying it.
Let \(h_d(n)\) be the minimum worst-case query cost for monotone maps on
\([n]^d\) with a unique fixed point, when the algorithm is required to
query that point. An algorithm that outputs the point can meet this
requirement with one final query. Hence
\begin{equation}\label{eq:hitting}
U_d(n)\le h_d(n)\le U_d(n)+1.
\end{equation}
This distinction is useful because our construction makes a comparison
sign available only when a local fixed point is queried.

\begin{theorem}[Dimension amplification]\label{thm:amplification}
For integers \(d\ge1\), \(n\ge3\), and \(M\ge2\),
\begin{equation}\label{eq:amplification}
U_{4d+1}(Mn)\ge
\frac{h_d(n)}4\ceil{\log(M-1)}
\ge\frac{U_d(n)}4\ceil{\log(M-1)}.
\end{equation}
\end{theorem}

Here \(n\) is the side length of each local instance, and \(M\) is the
number of local instances used in the construction. Taking \(M=n\)
squares the side length and contributes a factor of order \(\log n\).
Starting with the planar lower bound and iterating gives
Theorem~\ref{thm:hierarchy}. The dimension increase and the factor of
four in the query simulation are explained below.

\subsection{Proof overview}

The amplification theorem combines a product lower bound for concealed
ordered search with a monotone representation of that search. The
representation's unique fixed point reveals the answer, and each full
function evaluation uses at most four auxiliary queries.

\paragraph{Making every comparison costly.}
Consider an unknown threshold \(J\in[M-1]\). At index \(i\), the
comparison sign is \(+1\) if \(i\le J\) and \(-1\) otherwise.
Instead of revealing this sign directly, attach it to the unique fixed
point of an independent map on \([n]^d\). A source query returns the
full value of one local map, together with the sign only if the query
hits its fixed point. Thus the maps hide \emph{where} comparisons can be
read, while the threshold determines \emph{what} those comparisons say.

A threshold adversary can force \(\lceil\log(M-1)\rceil\) local fixed
points to be queried. A separate adversary for each local map forces at
least \(h_d(n)\) queries before that map's fixed point is reached.
These statements must remain valid when an algorithm interleaves work
on many maps; assuming that it finishes one local search before starting
another would not suffice. Section~\ref{sec:search-lower} proves the
product bound by maintaining the remaining worst-case query cost of each
local problem and the set of thresholds consistent with all revealed signs.
This product bound also follows from Montanaro's deterministic
decision-tree composition theorem \cite{montanaro}.

\paragraph{Encoding a concealed sign monotonically.}
A fixed point of the final map should identify the threshold, but a query
away from the relevant local fixed points must be answerable without the
concealed signs. Moreover, the entire answer vector must be monotone.
To reconcile these requirements, pair a local map with its reflected copy.
If its fixed point is \(p\), the pair has fixed point
\((p,(n+1)\one_d-p)\), where \(\one_d=(1,\ldots,1)\).
All such points have the same coordinate sum and are pairwise
incomparable. We can therefore place either sign at the unknown fixed
point, assign \(-1\) below this known level and \(+1\) above it, and
use a fixed default sign at all other points on the level. This gives a
monotone label that needs the concealed sign only where the oracle
already reveals it.

\paragraph{Combining adjacent comparisons.}
Two adjacent signs locate the threshold: their average is positive
before the pair \(J,J+1\), zero at that pair, and negative afterwards.
We use two groups of \(2d\) coordinates to hold adjacent reflected
instances and one final coordinate to select the pair. Consecutive
pairs retain one instance and replace the other. Translating local
ranges into ordered boxes makes the spatial outputs increase at each
replacement. In the last coordinate, changing directly from one sign
to the next could violate monotonicity; inserting their average spreads
a possible drop of two over two unit steps. This explains both the
three positions per pair and the dimension \(2d+2d+1=4d+1\).

The resulting map is monotone and has one fixed point, whose last
coordinate identifies \(J\). Evaluating either reflected instance takes
at most two source queries, so the entire target value takes at most
four. The construction accepts any monotone local map with a unique
fixed point. Its output can therefore be
used as a local instance in the next application.

\subsection{Related work}

\paragraph{Upper bounds and the role of dimension.}
For each fixed dimension \(d\), Dang, Qi, and Ye \cite{dqy} gave the
recursive binary-search bound \(O((\log N)^d)\). Fearnley, P\'alv\"olgyi, and Savani \cite{fps}
showed that three dimensions can be handled in \(O((\log N)^2)\)
queries and gave a decomposition theorem for higher dimensions.
Chen and Li \cite{chenli} subsequently obtained
\[
O\bigl((\log N)^{\lceil(d-1)/2\rceil+1}\bigr)
\]
queries. Chen, Li, and Yannakakis \cite{cly2026} improved this to
\[
O\bigl((\log N)^{\lceil(d-1)/3\rceil+1}\bigr),
\]
in particular obtaining the quadratic-logarithmic bound in dimension
four. Haslebacher and Lill \cite{hl} gave a different algorithm with
the same bound in dimension three. These improvements explain why
nested binary search alone does not establish that the exponent must
grow with dimension.

\paragraph{Path constructions and search composition.}
Etessami et al.\ \cite{epry} proved an \(\Omega((\log N)^2)\)
randomized lower bound in dimension two using a monotone path
construction with a unique fixed point. This is our starting bound;
Appendix~\ref{app:planar} gives a self-contained deterministic proof
using that path map. Br\^anzei, Phillips, and Recker \cite{bpr}
extended the construction to multidimensional herringbone functions
and proved that, for \(d,N\ge2\), randomized algorithms need an expected
number of queries of at least
\(\Omega((d/\log d)(\log N)^2)\).
This bound allows any fixed success probability greater than one half.
Their deterministic \(O(d\log N\log(Nd))\)-query algorithm for that
family also shows that its complexity remains quadratic in \(\log N\)
for each fixed \(d\).
In the quantum query model, Li, Ma, Yang, and Zhao \cite{lmyz}
proved an \(\Omega(d\log N)\) lower bound on \([N]^d\).
For fixed \(d\), this bound is linear in \(\log N\).

There is a related composition principle in Phillips's work on quantum
Tarski queries \cite[Definition~4 and Theorem~2]{phillips}.
His generalized-search framework allows
answers to be concealed at solutions of other search problems, and he
represents nested ordered search by two-dimensional monotone maps with a
unique fixed point, obtaining a quadratic-logarithmic quantum lower bound.
His representation also has constant evaluation cost
\cite[Lemma~13]{phillips}. Here we give a monotone
representation that accepts arbitrary unique-fixed-point instances as its
local maps and increases their dimension from \(d\) to \(4d+1\), allowing
the logarithmic factor to be gained repeatedly.

\paragraph{Uniqueness and the scope of the reduction.}
Chen, Li, and Yannakakis \cite{clyunique} proved that general and unique
Tarski have exactly the same deterministic query complexity when an
algorithm must query a fixed point. Their reduction constructs answers
adaptively from the query history; the output convention used here differs
by at most one query, as in \eqref{eq:hitting}. Our construction instead
defines a fixed target map for each fixed auxiliary input. It is not a
standard reduction from unrestricted Tarski to unique Tarski: its source
oracle also carries comparison signs and promises that they arise from
one threshold. This distinction matters for the standard-reduction
obstruction in \cite{clyunique}. Appendix~\ref{app:reductions} explains
why compatible partial query records need not have a common completion
within our promised source domain.

\paragraph{Organization.}
Section~\ref{sec:model} specifies the model and query-complexity
conventions. Section~\ref{sec:search} derives the dimension-amplification
theorem from an auxiliary search lower bound and a monotone
representation. Section~\ref{sec:construction} constructs that
representation, and Section~\ref{sec:search-lower} proves the auxiliary
lower bound. Section~\ref{sec:consequences} iterates the construction to
obtain the lower-bound hierarchy, and Section~\ref{sec:discussion}
discusses the scope of the result.

\section{Model}\label{sec:model}

\subsection{Grids and input families}

\begin{definition}[Monotone fixed-point instances]\label{def:instances}
For integers \(d,N\ge1\), let \([N]=\{1,\ldots,N\}\). For
\(x,y\in[N]^d\), define
\[
x\preceq y\quad\Longleftrightarrow\quad
x_i\le y_i\ \text{for every }i\in[d].
\]
A map \(F:[N]^d\to[N]^d\) is \emph{monotone} if
\[
\forall x,y\in[N]^d,\qquad
x\preceq y\ \Longrightarrow\ F(x)\preceq F(y).
\]
Its set of fixed points is
\(\Fix(F)=\{x\in[N]^d:F(x)=x\}\). Define
\[
\begin{aligned}
\mathcal M_{d,N}
&=\{F:[N]^d\to[N]^d:F\text{ is monotone}\},\\
\HH_{d,N}
&=\{F\in\mathcal M_{d,N}:|\Fix(F)|=1\}.
\end{aligned}
\]
The general Tarski problem has input family \(\mathcal M_{d,N}\);
the unique Tarski problem has input family \(\HH_{d,N}\).
In either case, a valid output on input \(F\) is a point in \(\Fix(F)\).
\end{definition}

Tarski's theorem ensures that \(\Fix(F)\ne\varnothing\) for every
\(F\in\mathcal M_{d,N}\). Monotonicity, and uniqueness when applicable,
are input promises; the algorithm need not verify them. We use the
coordinatewise order on all product grids and write \(\one_k\) for
the all-ones vector in \(k\) coordinates.

\subsection{Oracle algorithms and success conventions}

\begin{definition}[Deterministic full-value query algorithm]\label{def:algorithm}
Fix \(d,N\ge1\) and a nonempty input family
\(\mathcal C\subseteq\mathcal M_{d,N}\), known to the algorithm.
The unknown input \(F\in\mathcal C\) is fixed throughout the execution.
A query at \(x\in[N]^d\) returns the entire vector \(F(x)\) and costs
one query, independently of \(d\).

A deterministic algorithm is a rule that, given \(d,N\), the family
\(\mathcal C\), and a transcript of \(t\ge0\) queries
\[
\mathcal T_t=((x_1,y_1),\ldots,(x_t,y_t)),\qquad y_s=F(x_s),
\]
either chooses a next query \(x_{t+1}\in[N]^d\) or terminates with an
output \(z\in[N]^d\). The initial transcript \(\mathcal T_0\) is empty.
There is no restriction on computation between queries, and the rule
has no access to \(F\) other than through its query replies.

Let \(\mathcal A_{\mathrm{out}}(\mathcal C)\) be the algorithms that,
for every \(F\in\mathcal C\), terminate after finitely many queries
and output a point of \(\Fix(F)\). The output need not have been queried.
Let \(\mathcal A_{\mathrm{hit}}(\mathcal C)\) be the subclass for which,
on every \(F\in\mathcal C\), the transcript at termination contains
some pair \((x_s,y_s)\) with \(y_s=x_s\). Thus the hitting convention
requires an actual query at a fixed point, even if its location can be
deduced earlier. For either class, let \(T_A(F)\) denote the total number
of queries made by algorithm \(A\) before termination on input \(F\).
\end{definition}

For an arbitrary transcript \(\mathcal T_t\), its consistent input family is
\[
\mathcal C(\mathcal T_t)
=\{F\in\mathcal C:F(x_s)=y_s\text{ for }1\le s\le t\}.
\]
An adversary may choose replies as queries arrive, but must maintain
\(\mathcal C(\mathcal T_t)\ne\varnothing\) after every reply. This
condition ensures that the complete finite record has a legal fixed
input realizing all its answers.

\subsection{Query complexities}

\begin{definition}[Worst-case deterministic query complexity]\label{def:complexities}
For \(d,N\ge1\), define
\[
\begin{aligned}
Q_d(N)
&=\min_{A\in\mathcal A_{\mathrm{out}}(\mathcal M_{d,N})}
  \max_{F\in\mathcal M_{d,N}} T_A(F),\\
U_d(N)
&=\min_{A\in\mathcal A_{\mathrm{out}}(\HH_{d,N})}
  \max_{F\in\HH_{d,N}} T_A(F),\\
h_d(N)
&=\min_{A\in\mathcal A_{\mathrm{hit}}(\HH_{d,N})}
  \max_{F\in\HH_{d,N}} T_A(F).
\end{aligned}
\]
\end{definition}

The input families are finite, and querying every grid point gives a
correct algorithm under either convention. Hence these minima are finite
and attained. Restricting the input family gives \(U_d(N)\le Q_d(N)\).
Every hitting algorithm is also an output algorithm, while an output
algorithm can query its returned point once more. This proves
\(U_d(N)\le h_d(N)\le U_d(N)+1\), as stated in \eqref{eq:hitting}.
Requiring a hitting algorithm to terminate at its first fixed-point query
does not change \(h_d(N)\): it can stop at that query and return the
queried point. Our final lower bounds concern \(U_d\) and hence also \(Q_d\).

The auxiliary oracle in Section~\ref{sec:aux-problem} returns an additional
sign when a local fixed point is queried. Deducing the point's location
need not reveal that sign, which is why the auxiliary lower bound uses
\(h_d\). This extra information is specific to the auxiliary problem;
the Tarski oracle returns only the function value.

All logarithms are to base two. In a bound stated for fixed \(d\), the
implicit constants and the side-length threshold in \(O\), \(\Omega\),
or \(\Theta\) may depend on \(d\), but not on \(N\). The same convention
applies to a fixed iteration count \(r\) in the lower-bound hierarchy.

\section{Amplifying the query lower bound}\label{sec:search}

We now show how combining local fixed-point searches produces a harder
search problem. A local instance is a monotone map on \([n]^d\) with a
unique fixed point; write \(M\) for the number of these instances. The
dimension-amplification theorem (Theorem~\ref{thm:amplification})
transfers their search cost to a single instance on \([Mn]^{4d+1}\):
\[
U_{4d+1}(Mn)\ge \frac{h_d(n)}4\ceil{\log(M-1)}.
\]
Here \(h_d(n)\) is the worst-case cost of actually querying one local
fixed point. Thus the construction multiplies this cost by a logarithmic
factor, up to the constant loss of four. Taking \(M=n\) will increase
the power of the logarithm in the lower bound.

We prove this inequality through an auxiliary threshold-search problem.
The concealed ordered-search lemma (Lemma~\ref{lem:search}) shows that
hiding each comparison at a local fixed point forces the local query
cost and the number of comparisons to multiply. The monotone
representation (Proposition~\ref{prop:embedding}) then transfers this
cost to Tarski: it gives a map with a unique fixed point that identifies
the threshold, and evaluates each full target value using at most four
auxiliary queries. We state these two results and combine them here;
their proofs occupy the next two sections.

\subsection{The auxiliary search problem}\label{sec:aux-problem}

The auxiliary problem is ordered search with costly access to
comparisons. An unknown threshold divides the indices into two groups.
At each index, the oracle returns the comparison with that threshold
only when the algorithm queries the unique fixed point of the local map
assigned to the index. The required output is \(J\). This extra comparison information belongs to
the auxiliary oracle, not to the Tarski oracle defined in
Section~\ref{sec:model}.

Fix \(d\ge1\), \(n\ge3\), and \(M\ge2\). Recall that \(\HH_{d,n}\) is the finite family
of monotone maps \([n]^d\to[n]^d\) with a unique fixed point
(Definition~\ref{def:instances}). The answer
to the search problem is an unknown threshold \(J\in[M-1]\). For each
index \(i\in[M]\), choose a map \(H_i\in\HH_{d,n}\), and write \(p_i\) for
its fixed point. Every combination of a threshold and a tuple of maps is
a legal input: the choice of \(J\) imposes no restriction on any \(H_i\).
Define the comparison signs by
\begin{equation}\label{eq:bits}
b_i=\begin{cases}+1,&i\le J,\\-1,&i>J.\end{cases}
\end{equation}
Thus the signs change from \(+1\) to \(-1\) between indices \(J\) and
\(J+1\), and the endpoint signs \(b_1=+1\), \(b_M=-1\) are known.
A sign is revealed only at the corresponding fixed point: the source oracle
answers a query \((i,x)\in[M]\times[n]^d\) by
\begin{equation}\label{eq:source-oracle}
\OO(i,x)=\left(H_i(x),
\begin{cases}b_i,&H_i(x)=x,\\ \bot,&H_i(x)\ne x,\end{cases}\right),
\end{equation}
where \(\bot\) means that no comparison sign is returned. The required
output is \(J\). Ordinary function values do
not constrain the threshold, because the maps can be chosen independently
of it. Let \(S_d(M,n)\) be the minimum worst-case number of source queries
used by a deterministic algorithm that determines \(J\).

Locating a threshold requires logarithmically many comparisons in the
worst case. Here, each new comparison is available only after a local
search reaches its fixed point. The next lemma proves that these costs
multiply even when the algorithm interleaves queries to different maps.

\begin{lemma}[Concealed ordered search]\label{lem:search}
\[
S_d(M,n)\ge h_d(n)\ceil{\log(M-1)}.
\]
\end{lemma}

Lemma~\ref{lem:search} is a special case of Montanaro's composition
theorem \cite[Corollary~3]{montanaro}. The inner problem is to recover a
sign hidden at a local fixed point, which has query complexity \(h_d(n)\);
the outer problem is ordered search over \(M-1\) possible thresholds.
Section~\ref{sec:search-lower} gives a self-contained adversary proof in
the present oracle model, maintaining the work still required to reach
each local fixed point and the thresholds consistent with the revealed signs.
For \(M=2\), the threshold is already known and the bound is zero.

\subsection{Representing the search by a monotone map}

The auxiliary search lower bound must now be transferred to an ordinary
Tarski instance. We construct a monotone map whose unique fixed point
identifies the hidden threshold \(J\). Each target query must be
simulated with a constant number of source queries, so that the
simulation preserves the extra logarithmic factor. The next proposition
achieves this with at most four source queries per full target value.

\begin{proposition}[Monotone representation]\label{prop:embedding}
For fixed \(d\ge1\), \(n\ge3\), and \(M\ge2\), a single oracle procedure
evaluates a map
\[
F:[Mn]^{4d+1}\to[Mn]^{4d+1}
\]
using at most four queries to \(\OO\) for each target point. For every
auxiliary input represented by \(\OO\), this map is monotone and has a
unique fixed point whose last coordinate is \(3J-1\).
\end{proposition}

For a fixed auxiliary input, the target map is fixed as well: its value
at a point does not depend on the preceding target queries. The procedure
must return all \(4d+1\) output coordinates. Section~\ref{sec:construction} establishes these properties
by an explicit formula.

\subsection{Transferring the lower bound}

The search lower bound and the monotone representation together give
the desired amplification. The first charges the cost of a local
fixed-point search for each required threshold comparison. The second
turns a \(q\)-query target algorithm into an algorithm that determines
the threshold with at most \(4q\) source queries. Comparing these costs
proves the dimension-amplification theorem.

\begin{proof}[Proof of Theorem~\ref{thm:amplification}]
Suppose a deterministic algorithm solves every unique-fixed-point instance
on \([Mn]^{4d+1}\) using at most \(q\) queries. Given an auxiliary input,
simulate the algorithm using the monotone representation
(Proposition~\ref{prop:embedding}). The simulation uses at most \(4q\)
source queries. A returned fixed point with last coordinate \(t^*\)
determines \(J=(t^*+1)/3\), even if that point was not queried. The
concealed ordered-search bound (Lemma~\ref{lem:search}) gives
\[
4q\ge h_d(n)\ceil{\log(M-1)}.
\]
Minimizing over algorithms and using \eqref{eq:hitting} proves
\eqref{eq:amplification}.
\end{proof}

\section{Constructing the monotone map}\label{sec:construction}

We construct the monotone map needed to transfer the auxiliary search
lower bound to Tarski. Its unique fixed point must identify the hidden
threshold, while its value at any query point must be computable using
at most four source queries. We first prepare the local maps, then
combine adjacent instances, and finally verify these properties
(Proposition~\ref{prop:embedding}).

Recall the source input: a threshold \(J\in[M-1]\) and monotone maps
\(H_i:[n]^d\to[n]^d\), each with a unique fixed point \(p_i\).
The comparison sign \(b_i\) is \(+1\) for \(i\le J\) and \(-1\)
otherwise; the source oracle returns it only when \(p_i\) is queried.
Set \(L=Mn\). The target map will have domain \([L]^{4d+1}\), and
its fixed point will have last coordinate \(3J-1\).

We arrange that a target fixed point solves two local instances at
adjacent indices \(j,j+1\), where \(1\le j<M\). At their fixed points,
the average
\[
\frac{b_j+b_{j+1}}2
\]
is positive for \(j<J\), zero for \(j=J\), and negative for \(j>J\).
A final coordinate will use this test to identify the pair \(J,J+1\).
The difficulty is evaluating that coordinate away from local fixed points,
where the signs have not been revealed, while keeping the whole map
monotone. We resolve this locally with monotone labels and then combine
two labelled instances.

The notation records these successive transformations: \(H_i\) is the
original \(d\)-dimensional map, \(G_i\) its reflected pairing, \(A_i\)
its translation to the common spatial domain, and \(\sigma_i\) its sign
label. The corresponding fixed points will be denoted by \(p_i\),
\(z_i\), and \(\zeta_i\), respectively.

\subsection{Monotone labels for the local instances}\label{sec:label}

Each local map \(H_i\) carries a comparison sign \(b_i\), which the
source oracle returns only at its unique fixed point. To incorporate
this sign into a monotone target map, we transform the local map and
equip it with a monotone \emph{label} taking values in \(\{-1,+1\}\).
The label must equal \(b_i\) at the transformed fixed point and be
computable without \(b_i\) elsewhere. We also place the transformed
maps in ordered range boxes, so that switching between instances later
will not decrease their spatial outputs.

Simply changing a label at an unknown fixed point need not preserve
monotonicity: the labels at comparable points could force its value.
We avoid this problem by placing all possible fixed points on one known
level of the coordinate sum. Distinct points on this level are
incomparable; such a set is an \emph{antichain}. Assigning \(-1\) below
the level and \(+1\) above it leaves every label on the level free.
Reflection gives the required constant sum.
Define \(R:[n]^d\to[n]^d\) by
\[
Rz=(n+1)\one_d-z.
\]
This map reverses the coordinatewise order and satisfies \(R(Rz)=z\).

\begin{lemma}[Reflection]\label{lem:reflection}
The map
\begin{equation}\label{eq:reflection}
G_i(u,v)=\bigl(H_i(u),R H_i(Rv)\bigr),\qquad (u,v)\in[n]^{2d},
\end{equation}
is monotone. Its unique fixed point is \(z_i=(p_i,Rp_i)\), and
\[
\sum_{k=1}^{2d}(z_i)_k=d(n+1).
\]
\end{lemma}

\begin{proof}
The first coordinate group is monotone. In the second, the two occurrences of
\(R\) reverse the order twice, so \(v\mapsto R H_i(Rv)\) is also monotone.
The fixed-point equations are \(H_i(u)=u\) and \(H_i(Rv)=Rv\). Uniqueness for
\(H_i\) forces \(u=p_i\) and \(Rv=p_i\). The coordinate-sum identity follows
from \(p_i+Rp_i=(n+1)\one_d\).
\end{proof}

Reflection places all possible fixed points on a known level, which
will let us define the labels. It does not yet order the outputs of
different instances. That ordering will be needed when the last coordinate
switches between local maps. We obtain it by translating their ranges
into disjoint, coordinatewise ordered boxes.
For any positive integer \(m\), write
\[
P_m(s)=\min\{m,\max\{1,s\}\}\qquad(s\in\mathbb Z),
\]
and apply \(P_m\) coordinatewise to vectors. This monotone projection lets
us extend each translated map to the common domain \([L]^{2d}\).
Put \(o_i=(i-1)n\) and define
\begin{equation}\label{eq:translated}
A_i(X)=o_i\one_{2d}+G_i\bigl(P_n(X-o_i\one_{2d})\bigr).
\end{equation}
Its range lies in the box
\[
B_i=[o_i+1,o_i+n]^{2d}.
\]
The map \(A_i\) is monotone and has the unique fixed point
\(\zeta_i=o_i\one_{2d}+z_i\). To see uniqueness, any fixed point must belong
to the range box \(B_i\); inside that box, the projection in
\eqref{eq:translated} does nothing.

The label is defined on this entire common domain. Its value depends on
the coordinate sum of the original input, measured relative to the known
sum at \(\zeta_i\). Write
\begin{equation}\label{eq:rank}
w_i(X)=\sum_{k=1}^{2d}X_k-\bigl(2do_i+d(n+1)\bigr),
\end{equation}
so \(w_i(X)\) is negative below the fixed-point level, zero on it,
and positive above it. This shifted coordinate sum is computed from \(X\) itself, not
from its clamped local coordinates; distinct comparable inputs therefore
have strictly different values of \(w_i\). We assign \(+1\) on the level \(w_i=0\),
except at the fixed point itself:
\begin{equation}\label{eq:label}
\sigma_i(X)=
\begin{cases}
-1,&w_i(X)<0,\\
+1,&w_i(X)>0,\\
b_i,&w_i(X)=0\text{ and }A_i(X)=X,\\
+1,&w_i(X)=0\text{ and }A_i(X)\ne X.
\end{cases}
\end{equation}

\begin{lemma}[Monotone labels]\label{lem:label}
The function \(\sigma_i:[L]^{2d}\to\{-1,+1\}\) is monotone and satisfies
\(\sigma_i(\zeta_i)=b_i\). The pair \((A_i(X),\sigma_i(X))\) can be evaluated
using at most two queries to \(\OO\).
\end{lemma}

\begin{proof}
If \(X\preceq Y\) and \(X\ne Y\), then \(\sum_k X_k<\sum_k Y_k\). Hence
distinct points on the level \(w_i=0\) are incomparable. Assigning either sign
to each point on this level, with \(-1\) below and \(+1\) above, gives a monotone
function. Lemma~\ref{lem:reflection} places \(\zeta_i\) on this level, proving
the asserted value at the fixed point.

To evaluate the map and its label, compute \((u,v)=P_n(X-o_i\one_{2d})\) and query
\(\OO(i,u)\) and \(\OO(i,Rv)\). The two full vector answers determine
\(G_i(u,v)\) and therefore \(A_i(X)\). The label needs a comparison sign only
when \(A_i(X)=X\) and \(w_i(X)=0\). In this case \(X=\zeta_i\), so
\(u=Rv=p_i\); the queries just made have already returned \(b_i\).
These answers suffice in every case of \eqref{eq:label}. If the two query
arguments coincide, only one source query is needed.
\end{proof}

The range boxes also ensure that changing from index \(i\) to \(i+2\)
increases every spatial output. If \(1\le i\le M-2\),
then, for every \(X,Y\in[L]^{2d}\) and every coordinate \(k\),
\begin{equation}\label{eq:ordered-boxes}
(A_i(X))_k\le in<(i+1)n+1\le(A_{i+2}(Y))_k.
\end{equation}
The comparison holds for arbitrary inputs and for any choices of the two
local maps. This uniform separation is what will allow one local map to
be replaced without comparing the maps themselves.

\subsection{Combining adjacent instances}

We now use two labelled local maps at a time to locate the hidden
threshold. At their fixed points, the average of two consecutive signs
is positive before the threshold, zero at it, and negative afterwards.
A final coordinate will use this average to select the correct pair.
The challenge is to switch from one adjacent pair to the next without
breaking monotonicity; we do so by retaining their shared instance and
replacing only its neighbor.

Each paired instance has \(2d\) coordinates. Use two groups of this size
and one final coordinate, writing a point of \([L]^{4d+1}\) as
\((X,Y,t)\), with \(X,Y\in[L]^{2d}\) and \(t\in[L]\).
Odd indices use the \(X\) group and even indices use the \(Y\) group.
The pairs \(1,2\), then \(2,3\), then \(3,4\) can consequently share these
two groups, replacing just one map between consecutive pairs. In terms
of the indices assigned to \((X,Y)\), the first transitions are
\((1,2)\), \((3,2)\), and \((3,4)\). The retained index carries the
same label across the transition, while the replaced index increases
by two.

The coordinate \(t\) chooses the pair. We call the positions assigned
to indices \(j,j+1\) \emph{stage \(j\)}. A stage is a part of the input
grid, not a round of an algorithm: queries may visit stages in any order.
The construction is therefore defined simultaneously on the whole grid.
For each stage \(j\in[M-1]\), let
\[
\alpha_j=\begin{cases}j,&j\text{ odd},\\j+1,&j\text{ even},\end{cases}
\qquad
\beta_j=\begin{cases}j+1,&j\text{ odd},\\j,&j\text{ even}.
\end{cases}
\]
Thus \(\alpha_j\) is the odd member and \(\beta_j\) the even member of
\(\{j,j+1\}\). Define, for every index \(i\),
\[
\lambda_i(X,Y)=
\begin{cases}\sigma_i(X),&i\text{ odd},\\\sigma_i(Y),&i\text{ even}.
\end{cases}
\]

Fix the spatial input \((X,Y)\) for a moment, and let
\(a=\lambda_j(X,Y)\), \(b=\lambda_{j+1}(X,Y)\) be the two labels in stage
\(j\). The last output will be \(t\) plus a correction, clamped to
\([1,L]\). Within stage \(j\), this correction passes from \(a\) to \(b\).
A direct switch could lower the correction from \(+1\) to \(-1\), more
than a unit increase in \(t\) can absorb. Inserting their average gives
three positions:
\[
\begin{array}{c|ccc}
t&3j-2&3j-1&3j\\[2pt]
\hline
\text{last output before clamping}&t+a&t+(a+b)/2&t+b
\end{array}.
\]
For example, when \(a=+1\) and \(b=-1\), all three displayed outputs
are \(3j-1\). The correction decreases, but the last output does not;
only the middle input has zero correction.
At the next stage, when \(j<M-1\), the correction starts with the label
of \(j+1\). That map stays in place, so replacing \(j\) by \(j+2\) does
not change the correction at the boundary. Inequality~\eqref{eq:ordered-boxes}
ensures that the corresponding spatial output increases.

To give the definition on all of \([L]\), put \(T=3(M-1)\). Since \(n\ge3\),
we have \(T<L\). For \(t\in[L]\), set
\[
\tau_t=\min\{t,T\},\qquad
j_t=\ceil{\tau_t/3},\qquad r_t=\tau_t-3(j_t-1).
\]
Here \(j_t\in[M-1]\) selects the stage and \(r_t\in\{1,2,3\}\) selects
one of its three corrections. Define
\begin{equation}\label{eq:delta}
\delta_t(X,Y)=
\begin{cases}
\lambda_{j_t}(X,Y),&r_t=1,\\
\bigl(\lambda_{j_t}(X,Y)+\lambda_{j_t+1}(X,Y)\bigr)/2,&r_t=2,\\
\lambda_{j_t+1}(X,Y),&r_t=3.
\end{cases}
\end{equation}
The value is always an integer in \(\{-1,0,1\}\). For \(t>T\), these
definitions keep the last stage and give \(\delta_t=\lambda_M\).
The combined map is
\begin{equation}\label{eq:F}
F(X,Y,t)=\bigl(
A_{\alpha_{j_t}}(X),\ A_{\beta_{j_t}}(Y),\ P_L(t+\delta_t(X,Y))
\bigr).
\end{equation}
The actual coordinate \(t\) occurs in the last component; its truncated value
\(\tau_t\) only chooses the stage. Thus the tail reuses the final local
pair while retaining the actual coordinate \(t\) in the last output. Equations
\eqref{eq:translated}, \eqref{eq:label}, and \eqref{eq:F} define a single
map for each auxiliary input, independently of any query transcript.

\subsection{The unique fixed point}

We first verify that a fixed point of the combined map identifies the
hidden threshold \(J\), and that there is exactly one such point. The
two spatial coordinate groups must be fixed by their active local maps,
so the corresponding labels equal the true comparison signs. The last
coordinate can then be fixed only at the middle position of the unique
adjacent pair whose signs cancel. We also check the boundary projection,
which could otherwise create a fixed point with nonzero correction.

\begin{lemma}[Unique fixed point]\label{lem:F-root}
The map \(F\) has exactly one fixed point. Its last coordinate is \(3J-1\).
\end{lemma}

\begin{proof}
For a fixed \(t\), the first \(4d\) fixed-point equations force
\begin{equation}\label{eq:spatial-roots}
X=\zeta_{\alpha_{j_t}},\qquad Y=\zeta_{\beta_{j_t}}.
\end{equation}
At these coordinates, the two labels in use are \(b_{j_t}\) and
\(b_{j_t+1}\).

For \(1<t<L\), the last coordinate is fixed precisely when \(\delta_t=0\).
The first and third positions of each stage have a correction equal to a single
sign, so this cannot happen there. The tail \(t>T\) has correction
\(b_M=-1\). At the middle position of stage \(j\), the correction vanishes
precisely when \(b_j+b_{j+1}=0\), which by \eqref{eq:bits} holds exactly for
\(j=J\).

At the two endpoints the projection requires a separate check. At \(t=1\)
the correction on \eqref{eq:spatial-roots} is \(b_1=+1\), so the last output
is 2. At \(t=L\), it is \(b_M=-1\), so the last output is \(L-1\).
Neither endpoint is fixed.

Finally \(2\le3J-1\le3M-4<L\). Taking \(t=3J-1\) and the coordinates
in \eqref{eq:spatial-roots} satisfies every fixed-point equation. These
coordinates are forced, proving uniqueness.
\end{proof}

\subsection{Monotonicity}

Having identified the unique fixed point, we must verify that the
combined map preserves coordinatewise order on the entire grid. We
check spatial inputs while holding the pair-selecting coordinate \(t\)
fixed, and then check \(t\) while holding the spatial inputs fixed.
Monotone local maps and labels handle the first check. Ordered range
boxes and the interpolation through a label, the average, and the next
label handle the second. Unlike uniqueness, monotonicity does not require
the signs to come from a single threshold.

\begin{lemma}[Monotonicity of the target map]\label{lem:F-monotone}
The map \(F:[L]^{4d+1}\to[L]^{4d+1}\) is monotone.
\end{lemma}

\begin{proof}
Fix \(t\). Each of the first two output groups is a monotone map. Each
\(\lambda_i\) is monotone in \((X,Y)\), as is the average of two such labels.
The last projection \(P_L\) is monotone as well. Thus \(F\) is monotone in the
spatial coordinates with \(t\) fixed.

Now fix \(X,Y\) and increase \(t\) by one. Within a stage the first two output
groups are unchanged. Between stages, one index changes from \(j\) to \(j+2\),
so its output increases coordinatewise by \eqref{eq:ordered-boxes}; the other
group stays unchanged. Both groups are constant for \(t\ge T\).

For the last output, put \(a=\lambda_j(X,Y)\) and \(b=\lambda_{j+1}(X,Y)\)
within a stage. The successive changes in its correction are
\[
\frac{a+b}{2}-a=\frac{b-a}{2},\qquad
b-\frac{a+b}{2}=\frac{b-a}{2}.
\]
Since \(a,b\in\{-1,+1\}\), each change is at least \(-1\). It follows that
\[
(t+1+\delta_{t+1})-(t+\delta_t)=1+\delta_{t+1}-\delta_t\ge0.
\]
Between consecutive stages, the correction equals \(\lambda_{j+1}(X,Y)\)
on both sides of the transition. It is also unchanged for \(t\ge T\).
Thus \(t+\delta_t\), and hence \(P_L(t+\delta_t)\), is nondecreasing in \(t\).

For two comparable full inputs, first increase their spatial coordinates at
fixed \(t\), then increase \(t\) at fixed spatial coordinates. The two
monotonicity statements give the required comparison.
\end{proof}

\subsection{Simulating a query}

The final requirement is to evaluate the entire target vector without
knowing the concealed signs in advance. A target query uses only two
local maps. For each one, the same at most two source queries determine
both its transformed function value and its label, as shown in the
monotone-label lemma (Lemma~\ref{lem:label}). Thus one target query
costs at most four source queries, independently of the dimension and
the number of local instances.

\begin{lemma}[Four-query simulation]\label{lem:simulation}
One full query to \(F\) can be simulated with at most four queries to the
source oracle \(\OO\).
\end{lemma}

\begin{proof}
The query's last coordinate determines the two indices and which correction
formula to use. Apply
Lemma~\ref{lem:label} to evaluate
\((A_{\alpha_{j_t}}(X),\sigma_{\alpha_{j_t}}(X))\) and
\((A_{\beta_{j_t}}(Y),\sigma_{\beta_{j_t}}(Y))\), using two source queries
for each pair. These values determine the full vector in \eqref{eq:F}.
\end{proof}

The unique fixed point in Lemma~\ref{lem:F-root}, monotonicity in
Lemma~\ref{lem:F-monotone}, and the four-query simulation prove
Proposition~\ref{prop:embedding}.

\section{The cost of the auxiliary search}\label{sec:search-lower}

The preceding construction transfers an auxiliary search lower bound
to Tarski. We now prove that lower bound. Recall that an unknown
threshold \(J\) defines the signs \(b_i\): positive for \(i\le J\)
and negative otherwise. A source query returns the full value
\(H_i(x)\), and also returns \(b_i\) when \(H_i(x)=x\). If
\(S_d(M,n)\) is the cost of determining the threshold and \(h_d(n)\)
is the cost of querying one local fixed point, the concealed
ordered-search lemma (Lemma~\ref{lem:search}) asserts
\[
S_d(M,n)\ge h_d(n)\ceil{\log(M-1)}.
\]
The proof uses only this source oracle, not the target map from
Section~\ref{sec:construction}.

It is not enough simply to multiply the cost of one hard local map by
the depth of binary search: an algorithm may distribute queries among
many maps before revealing any sign. Instead, we maintain a worst-case
residual cost separately for each map. The local adversaries control
when signs can be revealed; a separate threshold adversary controls how
much each revealed sign helps. Finally, we verify that all answers admit
a common legal input.

\begin{proof}[Proof of Lemma~\ref{lem:search}]
\emph{Local query cost.}
We first describe an adversary for a single map. A transcript \(\tau\) records the
queries made to that map and their full vector answers. Let \(\HH(\tau)\) be the
nonempty subfamily of \(\HH_{d,n}\) consistent with those answers. Until a fixed
point has been queried, let \(r(\tau)\) be the smallest worst-case number of further
queries needed to query the fixed point of an unknown member of \(\HH(\tau)\).
After a fixed point is queried, put \(r=0\).

These numbers are finite, since the domain is finite. Considering the first query in an optimal decision tree gives, at every transcript in which no fixed point has been queried,
\begin{equation}\label{eq:bellman}
r(\tau)=1+\min_x\max_{y\text{ consistent with }\tau}
r(\tau;x,y).
\end{equation}
Here \(r(\tau;x,y)\) is the residual cost after appending query \(x\) and
answer \(y\); it is zero when \(y=x\). Thus, for every choice of the next query \(x\),
the adversary can choose a consistent answer \(y\) such that
\begin{equation}\label{eq:slow}
r(\tau;x,y)\ge r(\tau)-1.
\end{equation}
Starting from \(r(\varnothing)=h_d(n)\), it follows that the first query to the
fixed point occurs no earlier than the \(h_d(n)\)-th query to this map. Even if
the fixed point can already be deduced from the transcript, the task of querying
it still requires one query.

Run this adversary separately for each \(H_i\) until its fixed point is queried.
The algorithm may switch freely between maps. Inequality~\eqref{eq:slow} holds
for every next query, including one selected using all the other transcripts,
so the cost of \(h_d(n)\) applies to every map whose fixed point is reached.
After that first visit, fix any map consistent with its transcript and use it
for subsequent answers.

\emph{Threshold uncertainty.}
To answer the comparison signs, maintain a set \(K\subseteq[M-1]\) of
possible thresholds, initially \(K=[M-1]\). When the fixed point of \(H_i\) is queried
for the first time, partition \(K\) into
\[
K_+=\{J\in K:J\ge i\},\qquad K_-=\{J\in K:J<i\}.
\]
Keep a larger part and return its sign. If the sign was already determined, one
part is empty and the candidate set is unchanged. Repeated visits return the
same sign. After fixed points of \(s\) distinct maps have been queried,
\begin{equation}\label{eq:candidates}
|K|\ge\ceil{\frac{M-1}{2^s}}.
\end{equation}

\emph{Consistency and total cost.}
Every remaining threshold can be combined with any choice of locally consistent
maps to produce a legal input with this transcript. Each local family is nonempty
by construction. At queries that did not reach a fixed point, the answer
\(\bot\) places no restriction on the threshold. At queries that did reach one,
the returned sign agrees with every \(J\in K\). Thus choosing one map from each
local family and attaching the signs of any \(J\in K\) reproduces all recorded
answers.

If \(|K|\ge2\), two legal inputs with different answers \(J\) have the same
transcript. A correct deterministic algorithm cannot stop. Consequently it must
reach fixed points of at least \(\ceil{\log(M-1)}\) distinct maps. Each costs at
least \(h_d(n)\) queries to its own oracle, and these sets of queries are
disjoint. Multiplying the two bounds proves the result. For \(M=2\) the stated lower bound is
zero and the conclusion is immediate.
\end{proof}

\section{Iterating the amplification theorem}\label{sec:consequences}

The amplification construction accepts arbitrary monotone local maps
with unique fixed points and produces another map with the same
properties. We can therefore apply it again to its own output family.
Starting from a quadratic-logarithmic lower bound in dimension two,
each application increases the logarithmic exponent by one. After
\(r\) applications, this gives the lower-bound hierarchy of
Theorem~\ref{thm:hierarchy}: for every fixed \(r\ge0\), dimension
\(d_r=(7\cdot4^r-1)/3\) has a lower bound of
\(\Omega((\log N)^{r+2})\), whose implicit constant may depend on \(r\).

For the induction step, take \(M=n\) in the dimension-amplification
theorem. This replaces dimension \(d\) by \(4d+1\) and side length
\(n\) by \(n^2\). Squaring the side length only doubles its logarithm,
so it changes the multiplicative constant, not the logarithmic exponent.
The planar bound supplies the base case, and a padding lemma extends
the result from square side lengths to every sufficiently large integer.

\begin{lemma}[Planar lower bound \cite{epry}]\label{lem:planar}
There are constants \(c_0>0\) and \(N_0\ge3\) such that
\(U_2(N)\ge c_0(\log N)^2\) for every integer \(N\ge N_0\).
\end{lemma}

Appendix~\ref{app:planar} supplies a self-contained proof, with
\(c_0=1/96\) and \(N_0=16\). To increase the side length, we use a monotone
extension that leaves the set of fixed points unchanged.

\begin{lemma}[Side-length padding]\label{lem:padding}
If \(L\le N\), then \(U_d(N)\ge U_d(L)\).
\end{lemma}

\begin{proof}
Given a monotone map \(F:[L]^d\to[L]^d\) with a unique fixed point, define
\(\widehat F(x)=F(P_L(x))\) on \([N]^d\). The extension is monotone and its range is
contained in \([L]^d\), so every fixed point belongs to that smaller box.
Within that box, \(P_L\) is the identity, and the fixed points of
\(\widehat F\) and \(F\) coincide. A query to \(\widehat F\) uses one query
to \(F\), giving the result.
\end{proof}

\begin{proof}[Proof of Theorem~\ref{thm:hierarchy}]
The case \(r=0\) is the planar quadratic-logarithmic lower bound
(Lemma~\ref{lem:planar}), since \(d_0=2\).
Assume the statement holds for \(r\), and note that
\(d_{r+1}=4d_r+1\). For an integer \(N\) sufficiently large, set
\(n=\floor{\sqrt N}\) and apply the dimension-amplification theorem with
\(M=n\). For \(n\ge\max\{3,N_r\}\),
\begin{align*}
U_{d_{r+1}}(n^2)
&\ge\frac{U_{d_r}(n)}4\ceil{\log(n-1)}\\
&\ge\frac{c_r}{8}(\log n)^{r+3},
\end{align*}
where we used \(\log(n-1)\ge\tfrac12\log n\) for \(n\ge3\).
By side-length padding (Lemma~\ref{lem:padding}), the same lower
bound holds at side length \(N\). For \(N\ge16\),
\[
n\ge\frac{\sqrt N}{2},\qquad \log n\ge\frac14\log N.
\]
Consequently the induction holds, for example, with
\[
c_{r+1}=\frac{c_r}{8\cdot4^{r+3}},\qquad
N_{r+1}=\max\bigl\{16,\max\{3,N_r\}^2\bigr\}.
\]
These constants depend only on \(r\). The closed formula for \(d_r\)
follows from \(d_0=2\) and \(d_{r+1}=4d_r+1\).
\end{proof}

\begin{proof}[Proof of Corollary~\ref{cor:no-exponent}]
Given a finite real number \(a\), choose a fixed integer \(r\ge0\) with
\(r+2>a\). In the fixed dimension \(d=d_r\),
\[
Q_d(N)\ge U_d(N)\ge c_r(\log N)^{r+2}
\]
for all sufficiently large \(N\). Hence
\(Q_d(N)/(\log N)^a\ge c_r(\log N)^{r+2-a}\to\infty\).
Since \(d=d_r\) is fixed, this rules out an \(O((\log N)^a)\) upper bound even if its
implicit constant is allowed to depend on \(d\).
\end{proof}

\section{Discussion}\label{sec:discussion}

The deterministic query complexity of Tarski fixed-point search cannot
be bounded by a common fixed power of \(\log N\) in all fixed
dimensions, even when the multiplicative constant may depend on the
dimension. Our lower-bound hierarchy proves this while preserving a
unique fixed point at every step. It does not give tight bounds in the
resulting dimensions or determine whether dimension nine is the smallest
dimension admitting a cubic-logarithmic lower bound. Reducing the
dimensional cost of reflection and of combining adjacent instances is
a natural direction for strengthening the result.

The amplification theorem is deterministic. Its proof chooses answers
against the next query of a deterministic algorithm; it does not provide
a distributional lower bound that can simply be substituted into a
randomized analysis. Extending the product bound to randomized algorithms
would require a separate argument. On the constructive side, the essential
property is already independent of the local family: any monotone map
with a unique fixed point can be used as a local instance. This is the
property that makes repeated amplification possible.

\section*{AI Disclosure}

Generative AI tools were used for literature searches, translation, and
language editing. They also assisted in the development of many parts of the
mathematical argument and exposition in this paper. The human authors have verified all
AI-generated content and take full responsibility for it.

\appendix
\section{The planar starting bound}\label{app:planar}

The iteration in Section~\ref{sec:consequences} starts from a
quadratic-logarithmic lower bound in dimension two. This appendix proves
that bound under our full-value query model, allowing an algorithm to
output an unqueried fixed point. We use the monotone path map of
Etessami et al.\ \cite{epry}; both the lower bound and the
unique-fixed-point property are known. The purpose is to supply a
self-contained base case for Lemma~\ref{lem:planar}.

We describe a staircase with independently hidden turn coordinates,
choose its target by a hidden threshold, and then define the path map.
Public intervals restrict which hidden coordinates can affect any one
query, allowing a full function value to be evaluated with at most three
source comparisons. A comparison adversary supplies the two logarithmic
factors. The construction and the cost analysis follow the same
representation-and-product pattern as the main proof.

\Needspace{8\baselineskip}
\subsection{A staircase with hidden coordinates}

As in the main construction, we hide threshold comparisons behind
separate searches. Here each local search is for an unknown integer,
rather than a fixed point. These hidden integers specify the turns of
a staircase path, and the threshold selects one turn as the target.
The geometry will let us recover the threshold from that target.

Fix integers \(M\ge2\) and \(n\ge2\). The source input consists of
\(2M\) independently chosen numbers
\[
\xi_1,\ldots,\xi_M,\eta_1,\ldots,\eta_M\in[n]
\]
and a threshold \(J\in[M-1]\). All choices are allowed: the input domain is
\([n]^{2M}\times[M-1]\). For \(1\le j\le M+1\), set
\(b_j=+1\) when \(j\le J\) and \(b_j=-1\) when \(j>J\).
In particular, \(b_1=+1\) and \(b_M=b_{M+1}=-1\) are public.
A source query chooses one hidden number and a point \(z\in[n]\), and
receives their three-way comparison. Upon equality, a query to \(\xi_j\)
also receives \(b_j\), while a query to \(\eta_j\) receives \(b_{j+1}\).
The source task is to determine \(J\).

The hidden numbers determine the geometry of the path, independently
of which turn will be the target. We use them as turn coordinates of
a staircase. Put
\(L=Mn+2\) and introduce the disjoint public intervals
\[
I_j=\{(j-1)n+2,\ldots,jn+1\},\qquad j\in[M].
\]
Set
\[
x_j=(j-1)n+1+\xi_j,\qquad
y_j=(j-1)n+1+\eta_j,
\]
and add the public endpoints
\(x_0=y_0=1\) and \(x_{M+1}=y_{M+1}=L\).
Both coordinate sequences are strictly increasing. For
\(j=0,\ldots,M\), let the path \(P\) go horizontally from
\((x_j,y_j)\) to \((x_{j+1},y_j)\), then vertically to
\((x_{j+1},y_{j+1})\), always in unit steps. Its target vertex is
\[
z^*=(x_J,y_J).
\]
The public interval containing the first coordinate of \(z^*\) determines
\(J\). Thus a fixed point identifies the source answer without requiring
any additional oracle query.

For \(1\le j\le M\), a path vertex strictly inside the vertical segment
at \(x_j\) precedes the target exactly when \(j\le J\). A vertex strictly
inside the horizontal segment at \(y_j\) precedes it exactly when
\(j<J\). Thus movement along these segments toward the target is governed
by \(b_j\) and \(b_{j+1}\), respectively. This explains the different
signs returned by equality queries to \(\xi_j\) and \(\eta_j\).
The next subsection extends this movement to a monotone map on the whole
grid.

\subsection{A monotone map with the prescribed target}

We turn the staircase into a monotone map whose only fixed point is
the selected target vertex. On the path, a point moves one step toward
the target. Off the path, it moves toward the path while preserving
its coordinate sum. These rules leave only the target fixed; the main check is that they
preserve coordinatewise order.

Index the vertices of \(P\) as \(v_2,v_3,\ldots,v_{2L}\), where the
subscript is the coordinate sum:
\((v_s)_1+(v_s)_2=s\). Each step is \((1,0)\) or \((0,1)\).
Write \(c_s=(v_s)_1\) for \(2\le s\le2L\), and
\(z^*=v_{s^*}\), so \(s^*=x_J+y_J\).
At a path vertex other than the target, the map moves one step toward the
target; the target itself stays fixed:
\begin{equation}\label{eq:path-on}
F(v_s)=
\begin{cases}
v_{s+1},&s<s^*,\\
v_s,&s=s^*,\\
v_{s-1},&s>s^*.
\end{cases}
\end{equation}
An off-path point moves toward the path on its line of constant coordinate
sum. For such a point \((a,b)\), put \(s=a+b\) and define
\begin{equation}\label{eq:path-off}
F(a,b)=
\begin{cases}
(a+1,b-1),&a<c_s,\\
(a-1,b+1),&a>c_s.
\end{cases}
\end{equation}
These values stay in \([L]^2\). If \(a<c_s\), then \(a<L\) and
\(b>s-c_s\ge1\); the other case is symmetric. Every off-path point
moves, and every path vertex except the target moves to a neighbor.
Thus \(z^*\) is the unique fixed point.

\paragraph{Monotonicity of the path map.}
It suffices to compare horizontally and vertically adjacent points.
The only nontrivial comparisons are those in which an input or its
neighbor lies on the path; away from the path, both points undergo
the same translation. Consider \(q=(a,b)\) and \(q'=(a+1,b)\) in the grid,
and write \(s=a+b\). Since each step of \(P\) is horizontal or vertical,
\(c_{s+1}-c_s\in\{0,1\}\).

If \(a<c_s\), then \(a+1\le c_{s+1}\). When the inequality is strict,
\(F(q')=F(q)+(1,0)\). In the equality case \(q'=v_{s+1}\), we have
\(c_s=c_{s+1}=a+1\), and hence \(F(q)=v_s\). Each possible value of
\(F(v_{s+1})\) is at least \(v_s\).
If \(a>c_s\), then \(a+1>c_{s+1}\), and again
\(F(q')=F(q)+(1,0)\).

It remains to consider \(q=v_s\). If also \(q'=v_{s+1}\), the images are
ordered because the map \(s\mapsto s+\operatorname{sgn}(s^*-s)\) is
nondecreasing. Here \(\operatorname{sgn}\) takes the values \(-1,0,1\)
according to the sign of its argument. If \(q'\) is off the path, then
\(c_{s+1}=a\) and
\(F(q')=(a,b+1)=v_{s+1}\). Formula~\eqref{eq:path-on} gives
\(F(q)\preceq v_{s+1}\). Exchanging the two coordinates proves the
vertical comparison, so \(F\) is monotone.

\paragraph{A row description for query evaluation.}
The definition above uses the path vertex of a given coordinate sum.
The simulation will instead locate the path in the query's row, so we
first justify that equivalent description.
In each row, the path occupies a nonempty interval of grid points,
possibly a singleton. Points strictly to its left use the first branch of
\eqref{eq:path-off}, and points strictly to its right use the second.
To see the first assertion, let \((u,b)\) be the first path vertex in the
row and suppose \(a<u\). The vertex \(v_{a+b}\) precedes \((u,b)\),
so its second coordinate is at most \(b\). Its first coordinate is
therefore at least \(a\); equality would put \((a,b)\) on the path, a
contradiction. Hence \(a<c_{a+b}\). The argument with the last path
vertex in the row proves the second assertion.

\subsection{Evaluating the map}

To transfer a lower bound from the hidden-integer search, a grid query
must be simulated at constant cost. We show that a full value
\(F(a,b)\) uses at most three source comparisons. The public intervals
containing the staircase turns let us locate the path in the query's
row; every threshold sign needed for the output is either public or
returned by one of these comparisons upon equality.

We now give the evaluation rule. For \(j\in[M]\), comparing \(a\) with \(x_j\) requires a
query only when \(a\in I_j\): outside the interval the answer is public,
and inside it we query \(\xi_j\) at \(a-((j-1)n+1)\).
Equality also supplies \(b_j\). Likewise, comparing \(b\) with \(y_j\)
costs at most one query, and equality supplies \(b_{j+1}\).
The artificial endpoints require no queries.

\paragraph{Interior rows away from a horizontal segment.}
First consider an interior row \(1<b<L\). The public intervals partition
\(\{2,\ldots,L-1\}\), so the row determines
\(k=\floor{(b-2)/n}+1\) with \(b\in I_k\).
Compare \(b\) with \(y_k\). If \(b<y_k\), the row intersects the path
only at \((x_k,b)\), strictly inside a vertical segment. If \(b>y_k\),
that point is instead \((x_{k+1},b)\), also strictly inside a vertical
segment. Both statements follow from
\[
y_{k-1}<\min I_k\le b\le\max I_k<y_{k+1},
\]
which also holds for \(k=1\) and \(k=M\) by the artificial endpoints.
In either vertical case, compare \(a\) with the corresponding \(x_j\).
A strict comparison determines the off-path value by the row description
above. Equality returns
\[
F(a,b)=(a,b+b_j).
\]
When \(j\le M\), the equality query to \(\xi_j\) supplies \(b_j\).
When \(j=M+1\), the segment is at \(a=L\) and \(b_{M+1}=-1\) is
public.

\paragraph{Rows containing a horizontal segment.}
If \(b=y_k\), the path occupies the horizontal interval
\([x_k,x_{k+1}]\) in that row. Compare \(a\) with its two endpoints.
Outside the interval, use \eqref{eq:path-off}. On the interval, the
equality \(b=y_k\) has already supplied \(b_{k+1}\), and the values
are as follows:
\begin{itemize}
\item If \(x_k<a<x_{k+1}\), return \((a+b_{k+1},b)\).
\item If \(a=x_{k+1}\), return \((a,b+1)\) when \(b_{k+1}=+1\),
and \((a-1,b)\) when \(b_{k+1}=-1\).
\item If \(a=x_k\), the equality query to \(\xi_k\) has also supplied
\(b_k\). Return \((a,b-1)\) when \(b_k=-1\), return \((a+1,b)\)
when \(b_{k+1}=+1\), and otherwise return \((a,b)\).
\end{itemize}
In the last case, the remaining sign pair is
\((b_k,b_{k+1})=(+1,-1)\), which means \(k=J\). The threshold promise
excludes the pair \((-1,+1)\). Thus these rules give exactly the
preceding vertex, following vertex, or target prescribed by
\eqref{eq:path-on}.

\paragraph{Boundary rows and query count.}
For the bottom row \(b=1\), compare \(a\) with \(x_1\). Return
\((a+1,1)\) if \(a<x_1\), \((a,2)\) if \(a=x_1\), and
\((a-1,2)\) if \(a>x_1\). The target follows the initial horizontal
segment because \(J\ge1\).
For the top row \(b=L\), return \((a+1,L-1)\) when \(a<L\), and
\((L,L-1)\) at \((L,L)\). The target precedes the final vertical
segment because \(J\le M-1\).

A full value \(F(a,b)\) therefore requires at most one comparison with a
\(y\)-coordinate and two with \(x\)-coordinates. Boundary queries
require no more. Every sign used is either public or supplied by an equality
query to its associated hidden coordinate.

\subsection{The query lower bound}

The source search has two costs: reaching a hidden integer to obtain a
threshold comparison, and obtaining enough comparisons to determine
the threshold. We prove that these costs multiply even when queries to
different integers are interleaved. Combining this product bound with
the three-comparison simulation then gives the planar Tarski lower bound.

Every deterministic algorithm for the source task requires at least
\begin{equation}\label{eq:planar-source}
\log(n+1)\ceil{\log(M-1)}
\end{equation}
queries in the worst case. As in Section~\ref{sec:search-lower}, the
proof combines a local adversary, a threshold adversary, and a common
completion argument.

\paragraph{The cost of an equality query.}
For each hidden number, maintain a nonempty interval of possible values. If the interval has
\(m\ge2\) elements, retain a larger nonempty strict side of the query
rather than return equality. At least \((m-1)/2\) values remain.
A query outside the interval leaves it unchanged. When the interval is a
singleton, equality is returned only when that singleton is queried;
the number then remains fixed at its revealed value.

After \(t\) nonequality queries to one number, at least
\((n+1)/2^t-1\) values remain. If equality is first returned on its
\(q\)-th local query, the preceding interval has size one, so
\[
1\ge\frac{n+1}{2^{q-1}}-1,
\qquad q\ge\log(n+1).
\]
This argument applies separately to every number, even when queries to
different numbers are interleaved arbitrarily.

\paragraph{The number of threshold signs required.}
Maintain also a set \(K\) of possible thresholds, initially \([M-1]\).
When equality is first obtained for a number carrying sign \(b_j\), retain
a larger of \(K\cap\{J\ge j\}\) and \(K\cap\{J<j\}\), and return
the corresponding sign. A previously determined sign leaves \(K\)
unchanged, including when it is revealed by a different number. Repeated
equality queries return the same sign. After equality has been obtained for
\(s\) distinct numbers,
\[
|K|\ge\ceil{(M-1)/2^s}.
\]
For every \(J\in K\), independently choosing each hidden number from
its remaining interval gives a legal input with exactly this transcript.
Nonequality answers contain no threshold sign, and all equality answers
agree with every retained threshold. Thus the numerical comparisons impose
no further restriction on \(J\), and the algorithm cannot stop while two
thresholds remain. It must obtain equality for at least
\(\ceil{\log(M-1)}\) distinct numbers, each costing at least
\(\log(n+1)\) queries. The queries for these numbers are disjoint,
proving \eqref{eq:planar-source}. When \(M=2\), the right-hand side is
zero and the bound holds trivially.

\paragraph{Transferring the bound and choosing the parameters.}
The unique fixed point of \(F\) identifies \(J\) through its first
coordinate. A \(q\)-query algorithm for this family, even one that outputs
an unqueried point, therefore gives a source algorithm using at most
\(3q\) queries. Hence
\begin{equation}\label{eq:planar-product}
U_2(Mn+2)\ge\frac13\log(n+1)\ceil{\log(M-1)}.
\end{equation}
For any integer \(N\ge16\), set \(m=\floor{\sqrt{N-2}}\), take
\(M=n=m\), and put \(L=m^2+2\le N\). We have \(m\ge3\) and
\[
m\ge\sqrt{N-2}-1\ge\sqrt N-2\ge\frac{\sqrt N}{2}.
\]
Consequently \(\log m\ge\tfrac14\log N\), and
\(\log(m-1)\ge\tfrac12\log m\). Applying Lemma~\ref{lem:padding} to
\eqref{eq:planar-product} gives
\[
U_2(N)\ge\frac13\log(m+1)\log(m-1)
\ge\frac16(\log m)^2
\ge\frac1{96}(\log N)^2.
\]
This proves Lemma~\ref{lem:planar}, with \(c_0=1/96\) and \(N_0=16\).

\Needspace{10\baselineskip}
\section{The threshold promise and standard reductions}\label{app:reductions}

Our reduction uses an auxiliary oracle with a global promise: its
comparison signs change from \(+1\) to \(-1\) at a single threshold.
It is not a reduction from unrestricted Tarski inputs to unique Tarski.
The promise ensures that the constructed map has just one fixed point,
and it explains why the standard-reduction obstruction of Chen, Li, and
Yannakakis \cite{clyunique} does not apply here.

The essential distinction concerns partial query records. Two records
may agree on every query they share yet require signs that cannot come
from one threshold. They then have no common legal completion. We first
illustrate this failure and then explain why the obstruction for
unrestricted source functions relies on the opposite property.

\paragraph{Where the promise is used.}
For the map \(F\) defined in \eqref{eq:F}, the threshold promise
\eqref{eq:bits} ensures that \(b_j+b_{j+1}=0\) at exactly one adjacent pair.
This is where the promise
enters the uniqueness proof for \(F\). Monotonicity uses the monotone labels
and the ordered range boxes, and continues to hold for other sign sequences.
The distinction can be seen in a small example.

\paragraph{Nonconflicting records without a legal joint completion.}
Take \(M=4\), fix \(p\in[n]^d\), and set \(H_i(x)=p\) for every \(i,x\).
Each map is monotone and has the unique fixed point \(p\).
For the legal threshold \(J=1\), evaluating \(F\) at its fixed point with
last coordinate \(t=2\) uses only the source queries \((1,p)\) and
\((2,p)\), whose signs are \(+1,-1\). For the legal threshold \(J=3\),
evaluation at the fixed point with \(t=8\) uses only \((3,p)\) and
\((4,p)\), again with signs \(+1,-1\).
The two query records have disjoint source positions, so none of their
answers conflict. Their union, however, requires the sign sequence
\[
(+1,-1,+1,-1),
\]
which has no threshold of the form \eqref{eq:bits}. Substituting this
sequence into \eqref{eq:F} still gives a monotone map, but now the middle
positions of all three pairs give fixed points, at \(t=2,5,8\).
Thus two individually legal records can be combined into consistent
oracle answers while admitting no common completion within the promised
source domain.

\paragraph{Why the standard-reduction obstruction is different.}
This common-completion issue is relevant to the standard-reduction theorem
of Chen, Li, and Yannakakis \cite[Theorem~6.1]{clyunique}.
Their source problem is total Tarski: the input is an arbitrary function
on the source grid, and a valid answer is either a fixed point or a pair
witnessing nonmonotonicity. A standard reduction must produce a monotone
target map with a unique fixed point for \emph{every} source function,
including nonmonotone ones. For each target point, its function value is
computed by a fixed decision tree of depth at most \(D\). If a target
fixed point can be decoded into a valid source answer using at most \(c\)
further source queries, their theorem gives a source algorithm using at
most \(D^2+c\) queries.

In that unrestricted source domain, any two query records that agree on
their common queries have a common completion. To see how uniqueness
constrains such records, consider a path in the evaluation decision tree
for a target point \(x\) whose output is \(x\). The answers along this path
certify that \(x\) is fixed. Two such paths for different target points
must disagree on some common source query: otherwise their answers could
be combined and completed to a source function whose target map has both
fixed points. These pairwise conflicts are what the quadratic query bound
exploits. In the example above, the threshold promise excludes the joint
completion, even though the two records have no conflicting answers.

\paragraph{What the four-query simulation establishes.}
The four-query simulation in Lemma~\ref{lem:simulation} uses the source
oracle \eqref{eq:source-oracle}. Its input consists of the independent maps
\(H_i\) and one threshold \(J\); a query returns the full value \(H_i(x)\)
and also returns \(b_i\) precisely when \(H_i(x)=x\).
Lemma~\ref{lem:search} proves the lower bound for this oracle, including its
fixed-point comparison signs and the threshold promise. Each evaluation of
\(F\) uses at most four such queries, and its unique fixed point determines
the source answer \(J\).

\end{document}